\documentclass[11pt,onecolumn,draft]{IEEEtran}
\usepackage{epsf,psfrag,amssymb,amsfonts,cite,subfigure,color,enumerate}
\newtheorem{theorem}{Theorem}
\newtheorem{example}{Example}

\newtheorem{lemma}{Lemma}
\newtheorem{definition}{Definition}

\def\psfancypar#1#2{\begingroup\def\par{\endgraf\endgroup\lineskiplimit=0pt}
               \setbox2=\hbox{\large\sc #2}
               \newdimen\tmpht \tmpht \ht2 \advance\tmpht by \baselineskip
               \font\hhuge=Times-Bold at \tmpht
               \setbox1=\hbox{{\hhuge #1}}
               \count7=\tmpht \count8=\ht1
               \divide\count8 by 1000 \divide\count7 by \count8 
               \tmpht=.001\tmpht\multiply\tmpht by \count7 
               \font\hhuge=Times-Bold at \tmpht
               \setbox1=\hbox{{\hhuge #1}}
               \noindent
                \hangindent1.05\wd1
               \hangafter=-2 {\hskip-\hangindent
               \lower1\ht1\hbox{\raise1.0\ht2\copy1}%
                \kern-0\wd1}\copy2\lineskiplimit=-1000pt}

\newcommand{\beq}{\begin{equation}}
\newcommand{\eeq}{\end{equation}}
\newcommand{\bqa}{\begin{eqnarray}}
\newcommand{\eqa}{\end{eqnarray}}
\newcommand{\bqn}{\begin{eqnarray*}}
\newcommand{\eqn}{\end{eqnarray*}}

\newcommand{\be}{\begin{enumerate}}
\newcommand{\ee}{\end{enumerate}}
\newcommand{\bi}{\begin{itemize}}
\newcommand{\ei}{\end{itemize}}
\newcommand{\bd}{\begin{description}}
\newcommand{\ed}{\end{description}}
\newcommand{\ba}{\begin{array}}
\newcommand{\ea}{\end{array}}
\newcommand{\bde}{\begin{definition}}
\newcommand{\ede}{\end{definition}}
\newcommand{\bex}{\begin{example}}
\newcommand{\eex}{\end{example}}

\def\boxit#1{\vbox{\hrule\hbox{\vrule\kern3pt
        \vbox{\kern3pt#1\kern3pt}\kern3pt\vrule}\hrule}}

\def\reals{ { {\rm  I \kern-0.15em R }  } }
\def\complex{ {\,{{\rm C} \kern-0.50em \raise0.20ex {  |}}\, }}

\def\0bf{{\bf 0}}
\def\1bf{{\bf 1}}
\def\2bf{{\bf 2}}
\def\3bf{{\bf 3}}
\def\4bf{{\bf 4}}
\def\5bf{{\bf 5}}
\def\6bf{{\bf 6}}
\def\7bf{{\bf 7}}
\def\8bf{{\bf 8}}
\def\9bf{{\bf 9}}

\def\ybf{{\bf y}}

\def\ybf{{\bf y}}

\def\Rbf{{\bf R}}

\def\Cmat{\mathcal{C}}

\def\Emat{\mathcal{E}}

\def\Mmat{\mathcal{M}}
\def\Nmat{\mathcal{N}}

\def\Qmat{\mathcal{Q}}
\def\Rmat{\mathcal{R}}

\def\Umat{\mathcal{U}}

\def\Xmat{\mathcal{X}}
\def\Ymat{\mathcal{Y}}

\def\Rxx{\Rbf_{\ssstyle X\kern-.1em X}}

\let\ssstyle=\scriptscriptstyle

\def\Kout{\setbox1=\hbox{\Huge\bf K}\hbox to
1.05\wd1{\hspace{.05\wd1}
}}
\def\Sout{\setbox1=\hbox{\Huge\bf S}\hbox to 1.05\wd1{\hspace{.05\wd1}
}}

\def\scalefig#1{\epsfxsize #1\textwidth}
\makeatletter

\newcommand{\Rmnum}[1]{\expandafter\@slowromancap\romannumeral #1@}
\makeatother
\title{\LARGE Capacity Bounds and Sum Rate Capacities of a Class of Discrete Memoryless Interference Channels}
\author{Fangfang Zhu and Biao Chen\footnotemark}
\begin{document}
\maketitle
\let\thefootnote\relax\footnotetext{F. Zhu and B. Chen are with the Department of Electrical Engineering and Computer Science, Syracuse University, Syracuse, NY 13244 USA (email: fashu@syr.edu, bichen@syr.edu). The material in this paper was presented in part at the Information Theory and Applications (ITA) Workshop, San Diego, CA, February 2011 and the IEEE International Symposium on Information Theory, Boston, MA, July 2012. This work was supported in part by the National Science Foundation under Award CCF-0905320.}
\begin{abstract}
This paper studies the capacity of a class of discrete memoryless interference channels where interference is defined analogous to that of Gaussian interference channel with one-sided weak interference. The sum-rate capacity of this class of channels is determined. As with the Gaussian case, the sum-rate capacity is achieved by letting the transceiver pair subject to interference communicate at a rate such that its message can be decoded at the unintended receiver using single user detection. It is also established that this class of discrete memoryless interference channels is equivalent in capacity region to certain degraded interference channels. This allows the construction of capacity outer-bounds using the capacity regions of associated degraded broadcast channels. The same technique is then used to determine the sum-rate capacity of discrete memoryless interference channels with mixed interference as defined in the paper. The obtained capacity bounds and sum-rate capacities are used to resolve the capacities of several new discrete memoryless interference channels. 
\end{abstract}

\section{Introduction}\label{sec:intro}
The interference channel (IC) models the situation where the transmitters communicate
with their intended receivers while generating interference to unintended
receivers. Despite decades of intense research, the capacity region of IC remains unknown except for a few
special cases. These include interference channels with strong and very strong interference \cite{Carleial:75IT, Sato:78IT1,Sato:81IT, Han&Kobayashi:81IT, Costa&ElGamal:87IT}; classes of deterministic and semi-deterministic ICs \cite{ElGamal&Costa:82IT, Chong&Motani:09IT}; and classes of discrete degraded ICs \cite{Benzel:79IT, Liu&Ulukus:08IT}.

There exists a strong parallel, both in terms of capacity region and capacity achieving encoding schemes, between two classes of interference channels: the discrete memoryless interference channel (DMIC) and the Gaussian interference channel (GIC). 

A DMIC as described in \ref{sec:DMIC} is characterized by its discrete input and output alphabets as well as the channel transition probability $p(y_1y_2|x_1x_2)$. On the other hand, a GIC, in its standard form, has its outputs expressed as
\begin{eqnarray}
Y_1&=&X_1+aX_2+Z_1,\label{eq:GIC_y1}\\
Y_2&=&bX_1+X_2+Z_2,\label{eq:GIC_y2}
\end{eqnarray}
where $a$ and $b$ are
the channel coefficients corresponding to the interference links,
$X_i$ and $Y_i$ are the transmitted and received signals, and the
channel input sequence $X_{i1}, X_{i2}, \cdots, X_{in}$ is subject
to the power constraint $\displaystyle\sum_{j=1}^{n}\Emat[X_{ij}^2]\leq
nP_i$, $i=1,2$, $Z_1$ and $Z_2$ are Gaussian noises with zero mean
and unit variance, independent of $X_1, X_2$. We describe below parallel capacity results between the two types of interference channels.
\begin{description}
\item[\bf Very Strong Interference]\hfill\\
Carleial \cite{Carleial:75IT} defined the very strong interference for a GIC in standard form as
\begin{eqnarray}
a^2&\geq& 1+P_1,\label{eq:GIC_verystrong1}\\
b^2&\geq& 1+P_2\label{eq:GIC_verystrong2}
\end{eqnarray} 
in Eqs.~(\ref{eq:GIC_y1}) and (\ref{eq:GIC_y2}). In this case, interference can be decoded first and subtracted from the received signals, resulting in interference-free signals for the intended receivers. This sequential decoding scheme under the very strong interference condition achieves the following rate region
\bqa\label{eq:GIC_verystrong_CR}
\Rmat(P_1,P_2) = \left\{(R_1,R_2)\left|\begin{array}{c}
  0\leq R_1\leq\frac{1}{2}\log(1+P_1)\\
  0\leq R_2\leq\frac{1}{2}\log(1+P_2)
\end{array}\right.\right\}.\eqa

This rate region is also a natural outer bound, hence is indeed the capacity region of the GIC under very strong interference, and is achieved with Gaussian input. For {\em Gaussian input}, the condition in (\ref{eq:GIC_verystrong1}) and (\ref{eq:GIC_verystrong2}) implies that 
\begin{eqnarray}
I(X_1;Y_1|X_2)&\leq& I(X_1;Y_2), \label{eq:DMIC_verystrong_1}\\
I(X_2;Y_2|X_1)&\leq& I(X_2;Y_1). \label{eq:DMIC_verystrong_2}
\end{eqnarray}
Sato in \cite{Sato:78IT1} imposes the above condition on a DMIC with the additional requirement that it hold for all product input and obtained the capacity region for a DMIC with very strong interference to be
\[
\label{eq:DMIC_verystrong_CR}
\Rmat = \left\{(R_1,R_2)\left|\begin{array}{c}
  0\leq R_1\leq I(X_1;Y_1|X_2)\\
  0\leq R_2\leq I(X_2;Y_2|X_1)
\end{array}\right.\right\}.
\]

Sato alluded in \cite{Sato:78IT1} that (\ref{eq:DMIC_verystrong_1}) and (\ref{eq:DMIC_verystrong_2}) hold for all product input may be too restrictive, i.e., ``This is a sufficient condition for the coincidence of the bounds, but may not be necessary." In \cite{Xu-etal:Globecom2010}, it was established indeed that for a DMIC, the very strong interference can be relaxed to be such that conditions (\ref{eq:DMIC_verystrong_1}) and (\ref{eq:DMIC_verystrong_2}) need to be satisfied only for input distribution achieving the boundary points of the capacity region. This simple generalization broadens the class of DMIC with very strong interference and is also consistent with the GIC counterpart - it was shown in \cite{Xu-etal:Globecom2010} that (\ref{eq:DMIC_verystrong_1}) and (\ref{eq:DMIC_verystrong_2}) may be violated with non-Gaussian input even if (\ref{eq:GIC_verystrong1}) and (\ref{eq:GIC_verystrong2}) are satisfied. 
\item[\bf Strong Interference] \hfill\\
Han and Kobayashi \cite{Han&Kobayashi:81IT} and Sato \cite{Sato:81IT} independently obtained the capacity region of a GIC under strong interference, i.e., when, $a\geq 1$ and $b\geq 1$ in Eqs.~(\ref{eq:GIC_y1}) and (\ref{eq:GIC_y2}) as the following
\bqa\label{eq:GIC_strong_CR}
\Rmat(P_1,P_2)=\left\{(R_1,R_2)\left|\begin{array}{c}
  0\leq R_1\leq\frac{1}{2}\log(1+P_1)\\
  0\leq R_2\leq\frac{1}{2}\log(1+P_2)\\
  R_1+R_2\leq \min\{\frac{1}{2}\log(1+P_1+aP_2), \frac{1}{2}\log(1+bP_1+P_2)\}
\end{array}\right.\right\}.
\eqa

Clearly, this capacity region coincides with that of a compound multiple-access channel (MAC) where both receivers are expected to decode both messages. Notice that in the case of $a^2\geq 1+P_1$ and $b^2\geq 1+P_2$, the sum rate bound in (\ref{eq:GIC_strong_CR}) is inactive thus (\ref{eq:GIC_strong_CR}) includes (\ref{eq:GIC_verystrong_CR}) as its special case. Nevertheless, to achieve (\ref{eq:GIC_strong_CR}) under the strong interference condition, joint decoding instead of sequential decoding  is required at each receiver.

In \cite{Sato:81IT} Sato also conjectured the condition as well as the capacity region of DMICs under strong interfernce, which was eventually proved by Costa and El Gamal in 1987 \cite{Costa&ElGamal:87IT}. The strong interference for a DMIC is referred to the condition that the inputs $X_1$ and $X_2$ and corresponding outputs $Y_1$ and $Y_2$ satisfy
\bqa
I(X_1;Y_1|X_2)\leq I(X_1;Y_2|X_2),\\
I(X_2;Y_2|X_1)\leq I(X_2;Y_1|X_1),\label{eq:dmic_strong_cond2}
\eqa
for all product probability distribution on $\Xmat_1 \times \Xmat_2$.

The corresponding capacity region was shown to be the union of the rate pairs $(R_1,R_2)$ satisfying
\bqa
\Rmat=\left\{(R_1,R_2)\left|\begin{array}{c}
  0\leq R_1\leq I(X_1;Y_1|X_2Q)\\
  0\leq R_2\leq I(X_2;Y_2|X_1Q)\\
  R_1+R_2\leq \min\{I(X_1X_2;Y_1|Q), I(X_1X_2;Y_2|Q)\}
\end{array}\right.\right\},
\eqa
where $Q$ is a time-sharing parameter of cardinality $4$, and the union is over all probability distributions of the form $p(q)p(x_1|q)p(x_2|q)p(y_1y_2|x_1x_2)$, with $p(y_1y_2|x_1x_2)$ specified by the channel. It was established in \cite{Xu-etal:Globecom2010} that the condition in (\ref{eq:DMIC_verystrong_1}) and (\ref{eq:DMIC_verystrong_2}) are consistent with the strong interference condition for a GIC. That is, for a GIC in stardard form, $a\geq 1$ and $b\geq 1$ is equivalent to (\ref{eq:DMIC_verystrong_1}) and (\ref{eq:DMIC_verystrong_2}) for all product input distribution for a GIC.
\end{description}

While the capacity region for the general GIC remains unknown, there have been recent progresses in characterizing the sum-rate capacity of certain GICs, including: GICs with one-sided weak interference \cite{Sason:04IT}, noisy interference  \cite{Shang-etal:09IT, Motahari&Khandani:09IT, Annapureddy&Veeravalli:09IT}, and mixed interference \cite{Motahari&Khandani:09IT}. This paper attempts to derive parallel sum-rate capacity results for DMICs with weak one-sided and mixed inference which complement existing parallel results in the strong interference regime. Our definitions of one-sided, weak, or mixed interference are motivated by properties associated with the corresponding Gaussian channels. Some of those definitions are intimately related to those introduced in \cite{Chong-etal:07IT} which studies the capacity region of the discrete memoryless Z-channel.

The rest of the paper is organized as follows. Section~\ref{sec: preliminaries} presents the channel model and relevant previous results. Section~\ref{sec:main} defines the DMIC with one-sided weak interference and derives its sum-rate capacity. We refer to those DMICs with one-sided interference as DMZIC (i.e., discrete memoryless Z interference channel) for ease of presentation. The equivalence between the DMIC with weak one-sided interference and the discrete degraded interference channel (DMDIC) is established which allows one to construct a capacity outer-bound for the DMZIC using the capacity region of the associated degraded broadcast channel. Several specific DMICs are studied in Section III whose capacities or capacity bounds are obtained. Section~\ref{sec:mixed} defines DMICs with mixed interference  and derives the sum capapcity for this class of channels. Section~\ref{sec:conclusion} concludes this paper.

\section{Preliminaries}\label{sec: preliminaries}
\subsection{Discrete Memoryless Interference Channels\label{sec:DMIC}}
A discrete interference channel is specified by its input alphabets $\Xmat_1$ and $\Xmat_2$,  output alphabets $\Ymat_1$ and $\Ymat_2$, and the channel transition matrices
\begin{eqnarray}
p(y_1|x_1x_2)&=&\sum_{y_2\in \Ymat_2}p(y_1y_2|x_1x_2),\label{eq:DMIC_condition1}\\
p(y_2|x_1x_2)&=&\sum_{y_1\in \Ymat_1}p(y_1y_2|x_1x_2).\label{eq:DMIC_condition2}
\end{eqnarray}

The discrete IC is said to be \textit{memoryless} if
\begin{equation}
p(y_1^ny_2^n|x_1^nx_2^n)=\prod_{i=1}^n p(y_{1i}y_{2i}|x_{1i}x_{2i}).
\end{equation}

A $(n,2^{nR_1},2^{nR_2},\lambda_1,\lambda_2)$ \textit{code} for a DMIC with independent information consists of two message sets $\Mmat_1 =\{1,2,\cdots,2^{nR_1}\}$ and $\Mmat_2=\{1,2,\cdots,2^{nR_2}\}$ for senders $1$ and $2$ respectively, two encoding functions
\bqn
f_1: \Mmat_1\rightarrow \Xmat_1^n,\hspace{.1in} f_2: \Mmat_2\rightarrow\Xmat_2^n,
\eqn
and two decoding functions
\bqn
\varphi_1: \Ymat_1^n\rightarrow \Mmat_1,\hspace{.1in} \varphi_2: \Ymat_2^n\rightarrow \Mmat_2.
\eqn
The average probabilities of error are defined as
\small
  \bqn \lambda_1\!\!\!\!\!\!&=&\!\!\!\!\!\!\frac{1}{|\Mmat_1||\Mmat_2|}\sum_{w_1=1}^{2^{nR_1}}\sum_{w_2=1}^{2^{nR_2}}Pr\{\varphi_1(\ybf_1)\neq w_1|W_1=w_1,W_2=w_2\},\\
  \lambda_2\!\!\!\!\!\!&=&\!\!\!\!\!\!\frac{1}{|\Mmat_1||\Mmat_2|}\sum_{w_1=1}^{2^{nR_1}}\sum_{w_2=1}^{2^{nR_2}}Pr\{\varphi_2(\ybf_2)\neq w_2|W_1=w_1,W_2=w_2\}.
  \eqn
  \normalsize

A rate pair $(R_1,R_2)$ is said to be \textit{achievable} for a DMIC if there exist a sequence of
$(2^{nR_1}, 2^{nR_2}, n, \lambda_1, \lambda_2)$ codes such that $\lambda_1,\lambda_2\rightarrow 0$ as $n\rightarrow \infty$.
The capacity region of a DMIC is defined as the closure of the set of all
achievable rate pairs.

\subsection{Existing Results for GICs}
Sason \cite{Sason:04IT} proved that the sum-rate capacity for GICs with one-sided weak interference ($a<1$ and $b=0$ in Eqs.~(\ref{eq:GIC_y1}) and (\ref{eq:GIC_y2})) is
\bqn
C_{\mbox{sum}}=\frac{1}{2}\log(1+P_2)+\frac{1}{2}\log\left(1+\frac{P_1}{1+aP_2}\right).
\eqn

This sum-rate capacity is achieved by letting the transceiver pair subject to interference communicate at a rate such that its message can be decoded at the unintended receiver using single user detection, and the interference-free transceiver pair communicate at the maximum rate. The GIC with one-sided interference is often referred to as the Gaussian Z interference channel (GZIC). 

Motahari and Khandani \cite{Motahari&Khandani:09IT} established that the sum-rate capacity for GICs with mixed interference ($a\leq 1$ and $b\geq 1$) is
\bqn
C_{\mbox{sum}}=\min\left\{\frac{1}{2}\log\left(1+\frac{P_1}{1+aP_2}\right),\frac{1}{2}\log\left(1+\frac{bP_1}{1+P_2}\right)\right\} +\frac{1}{2}\log(1+P_2).
\eqn

To achieve this sum-rate capacity, the transceiver pair subject to strong interference communicates at a rate as if there is no interference, while the transceiver pair subject to weak interference communnicates at a rate such that its message can be decoded at both receivers using single user detection. We attempt to extend these results to DMICs with appropriately defined one-sided weak interference and mixed interference. This extension will in turn allow us to solve the capacity of new DMICs.
\subsection{Useful Properties of Markov Chains}
The following properties of Markov chains are useful throughout the paper \cite{book:GraphicalModel}:
\begin{itemize}
  \item Decomposition: $X-Y-ZW\Longrightarrow X-Y-Z$;
  \item Weak Union: $X-Y-ZW\Longrightarrow X-YW-Z$;
  \item Contraction: $(X-Y-Z)$ and $(X-YZ-W)\Longrightarrow X-Y-ZW$.
\end{itemize}

\section{The DMIC with One-sided Weak Interference}\label{sec:main}
\subsection{Channel Model and Sum Rate Capacity}
\begin{definition}
For the DMIC defined in Section~\ref{sec:DMIC}, if for all $x_1$, $x_2$, $y_2$,
\bqa
p(y_2|x_2)=p(y_2|x_1x_2),\label{DMZIC_condition2}
\eqa
or equivalently,
\bqa
X_1- X_2 - Y_2\label{eq:def_Z}
\eqa
forms a Markov chain, this DMIC is said to have one-sided interference.
\end{definition}

Clearly, the Markov chain condition (\ref{eq:def_Z}) holds for the GIC with $b=0$ in (\ref{eq:GIC_y2}). As with the Gaussian case, we refer to the DMIC with one-sided interference as simply discrete memoryless Z interference channel (DMZIC). 
From the definition, it follows that $X_1$ and $Y_2$ are independent for all input distribution $p(x_1)p(x_2)$.

To define DMZIC with weak interference, we first revisit some properties of Gaussian ZIC with weak interference.
It is straightforward to show that a Gaussian ZIC with weak interference is equivalent in its capacity region to a degraded Gaussian ZIC
satisfying the Markov chain
\bqa\label{eq:weakcond2}
X_2- (X_1,Y_2)- Y_1.
\eqa
The proof is similar to that in \cite{Chong-etal:07IT} for a Gaussian Z channel instead of a Gaussian Z interference channel. This motivates us to define DMZIC with weak interference as follows.
\begin{definition}
A DMZIC is said to have \textit{weak interference} if the channel transition probability factorizes as
\bqa
p(y_1y_2|x_1x_2)=p(y_2|x_2)p'(y_1|x_1y_2),
\eqa
for some $p'(y_1|x_1y_2)$, or, equivalently, the channel is stochastically degraded.
\end{definition}

In the absence of receiver cooperation, a stochastically degraded interference channel is equivalent in its capacity to a physically degraded interference channel. As such, we will assume in the following that the channel is physically degraded, i.e., the DMZIC with weak interference admits the Markov chain $X_2-(X_1,Y_2)-Y_1$. 

The channel transition probability $p(y_1y_2|x_1x_2)$ for this class of channels factorizes as
\bqa\label{eq:DMZIC_channel_prob}
p(y_1y_2|x_1x_2)&=&p(y_2|x_1x_2)p(y_1|x_1x_2y_2)\nonumber\\
&=&p(y_2|x_2)p(y_1|x_1y_2).
\eqa
As a consequence, the following inequality holds
\begin{equation}
I(U;Y_2)\geq I(U;Y_1|X_1),\label{ieq:weak}
\end{equation}
for all input distributions $p(x_1)p(u)p(x_2|u)$. We note that this condition is indeed what is needed in establishing the sum-rate capacity of this channel and was used in \cite{Farsani&Marvasti:2012IT} to define the weak interference for DMZIC. The definition used in this paper, while stronger than necessary, is much more intuitive and easier to verify.

The above definition of weak interference leads to the following sum-rate capacity result.
\begin{theorem}\label{thm:DMZIC_CR}
The sum-rate capacity of a DMZIC with weak interference as defined above is
\bqa\label{eq:DMZIC_SC}
C_{\mbox{sum}}=\max_{p(x_1)p(x_2)}\{I(X_1;Y_1)+I(X_2;Y_2)\}.
\eqa
\end{theorem}
\begin{proof}
This sum-rate is achieved by two receivers decoding their own messages while treating any interference, if present, as noise.

For the converse, we have
\small
\bqn
n(R_1+R_2)-n \epsilon &\stackrel{(a)}{\leq}&\!\!\!\!\!\!I(X_1^n;Y_1^n)+I(X_2^n;Y_2^n)\\
&\stackrel{(b)}{=}&\!\!\!\!\!\!\sum_{i=1}^n \left(H(Y_{1i}|Y_1^{i-1})-H(Y_{1i}|Y_1^{i-1}X_1^n)+H(Y_{2i}|Y_2^{i-1})-H(Y_{2i}|Y_2^{i-1}X_2^n)\right)\\
&\stackrel{(c)}{\leq}&\!\!\!\!\!\!\sum_{i=1}^n \left(H(Y_{1i})-H(Y_{1i}|Y_1^{i-1}X_1^nY_2^{i-1})+H(Y_{2i}|Y_2^{i-1})-H(Y_{2i}|Y_2^{i-1}X_{2i})\right)\\
&\stackrel{(d)}{=}&\!\!\!\!\!\!\sum_{i=1}^n\left(H(Y_{1i})-H(Y_{1i}|X_1^nY_2^{i-1})+I(X_{2i};Y_{2i}|U_i)\right)\\
&\stackrel{(e)}{=}&\!\!\!\!\!\!\sum_{i=1}^n\left(H(Y_{1i})-H(Y_{1i}|X_{1i}Y_2^{i-1})+I(X_{2i};Y_{2i}|U_i)\right)\\
&=&\!\!\!\!\!\!\sum_{i=1}^n (I(U_iX_{1i};Y_{1i})+I(X_{2i};Y_{2i}|U_i))\\
&=&\!\!\!\!\!\!\sum_{i=1}^n (I(X_{1i};Y_{1i})+I(U_i;Y_{1i}|X_{1i})+I(X_{2i};Y_{2i}|U_i))\\
&\stackrel{(f)}{\leq}&\!\!\!\!\!\!\sum_{i=1}^n (I(X_{1i};Y_{1i})+I(U_i;Y_{2i})+I(X_{2i};Y_{2i}|U_i))\\
&=&\!\!\!\!\!\!\sum_{i=1}^n (I(X_{1i};Y_{1i})+I(U_iX_{2i};Y_{2i}))\\
&\stackrel{(g)}{=}&\!\!\!\!\!\!\sum_{i=1}^n (I(X_{1i};Y_{1i})+I(X_{2i};Y_{2i})),
\eqn
\normalsize
where $U_i\triangleq Y_2^{i-1}$ for all $i$, $(a)$ follows the Fano's Inequality, $(b)$ is from the chain rule and the definition of mutual information, $(c)$ is because of the fact that conditioning reduces entropy, and that $Y_{2i}$ is independent of any other random variables given $X_{2i}$, $(d)$ is due to the memoryless property of the channel and the fact that $Y_{1i}$ is independent of any other random variables given $X_{1i}$ and $Y_{2i}$, then $(X_{1,i}^{n},Y_{1i})-(X_1^{i-1},Y_2^{i-1})- Y_1^{i-1}$ forms a Markov chain. By the weak union property, the Markov chain $Y_{1i}-(X_1^{n},Y_2^{i-1})-Y_1^{i-1}$ holds; $(e)$ is because of the Markov chain $(X_{1}^{i-1},X_{1,i+1}^n)- (X_{1i},Y_2^{i-1})- Y_{1i}$. This can be established using the \textit{independence graph} \cite{Kramer:03IT}. Alternatively, we first note that the Markov chain
\bqn
(X_{1}^{i-1},X_{1,i+1}^n,Y_2^{i-1})- (X_{1i},Y_{2i})- Y_{1i}
\eqn
holds, since given $X_{1i}$ and $Y_{2i}$, $Y_{1i}$ is independent of $X_{1}^{i-1},X_{1,i+1}^n,Y_2^{i-1}$.
By the weak union property, the following Markov chain is obtained:
\bqn
(X_{1}^{i-1},X_{1,i+1}^n)- (X_{1i},Y_2^i)- Y_{1i}.
\eqn
The independence between $Y_2^n$ and $X_1^n$ gives the Markov chain
\bqn
(X_{1}^{i-1},X_{1,i+1}^n)- X_{1i}-Y_2^i.
\eqn
The above two Markov chains lead to the following Markov chain:
\bqn
(X_{1}^{i-1},X_{1,i+1}^n)- X_{1i}-(Y_{1i},Y_2^i)
\eqn
by the contraction property. Again, using the weak union property and then the decomposition property, we obtain the Markov chain
\bqn
(X_{1}^{i-1},X_{1,i+1}^n)- (X_{1i},Y_2^{i-1})- Y_{1i}
\eqn
as desired. Since $U_i$ and $X_{1i}$ are independent, then $p(x_1x_2u)=p(x_1)p(u,x_2)$, thus $(f)$ comes from (\ref{ieq:weak}). Finally, $(g)$ follows from the Markov chain $U_i- X_{2i}- Y_{2i}$. Finally, by introducing a time-sharing random variable $Q$, one obtains
\bqn
R_1+R_2&\leq& I(X_1;Y_1|Q)+I(X_2;Y_2|Q)+\epsilon\\
&\leq& \max_{p(x_1)p(x_2)}\{I(X_1;Y_1)+I(X_2;Y_2)\}+\epsilon.
\eqn
\end{proof}

\textit{Remark $1$:} From the strong interference condition (\ref{eq:dmic_strong_cond2}), it is perhaps tempted to define the condition for weak interference as 
\bqa
I(X_2;Y_1|X_1)\leq I(X_2;Y_2),\label{eq:DMZIC_weakextention}
\eqa
for all product input distribution on $\Xmat_1\times\Xmat_2$. Notice that the right-hand side is same as $I(X_2;Y_2|X_1)$ given that this is one-sided interference channel. 
The Markov chain (\ref{eq:weakcond2}) is a sufficient, but not necessary, condition for the mutual information condition (\ref{eq:DMZIC_weakextention}).
An example is provided in Appendix z\ref{apx:example} such that the mutual information condition holds but the Markov chain is not valid. This is different from that of the Gaussian case; it can be shown that the coefficient $a\leq 1$ in a Gaussian ZIC is a sufficient and necessary condition for (\ref{eq:DMZIC_weakextention}) to hold. It is yet unknown if condition (\ref{eq:DMZIC_weakextention}) is sufficient for the sum-rate capacity result (\ref{eq:DMZIC_SC}) to hold for the DMZIC.

\textit{Remark $2$:}  For a DMZIC with weak interference, an achievable rate region, $\Cmat$, is given
by the set of all nonnegative rate pairs $(R_1,R_2)$ that satisfy
\bqa
R_1&\leq&I(X_1;Y_1|U_2Q),\label{eq:achievable_start}\\
R_2&\leq&I(X_2;Y_2|Q),\label{eq:achievable_mid}\\
R_1+R_2&\leq&I(U_2X_1;Y_1|Q)+I(X_2;Y_2|U_2Q).\label{eq:achievable_end}
\eqa
where the input distribution factorizes as:\bqa
p(qu_2x_1x_2)=p(q)p(x_1|q)p(u_2|q)p(x_2|u_2,q).\label{eq:input_dist}
\eqa
Furthermore, the region remains invariant if we impose the constraints $\|\Qmat\|\leq 5$, $\|\Umat_2\|\leq \|\Xmat_2\|+3$.
This can be readily obtained from the achievable rate region of the general two-user IC \cite{Han&Kobayashi:81IT, Chong-etal:08IT}.
In the next lemma, we provide a simpler description for the above achievable rate region.
\begin{lemma}
The region $\Cmat$ is equivalent to the set of all rate pairs $(R_1,R_2)$ satisfying
\bqa
R_1&\leq& I(X_1;Y_1|U_2'Q),\label{eq:achievable_alter_start}\\
R_2&\leq& I(U_2';Y_1Q)+I(X_2;Y_2|U_2'Q)\label{eq:achievable_alter_end}.
\eqa
where the input distribution factorizes as (\ref{eq:input_dist}). Furthermore, the region remains invariant if we impose the constraints $\|\Qmat\|\leq 4$, $\|\Umat_2'\|\leq \|\Xmat_2\|+3$.
\end{lemma}
\begin{proof}
Let $E$ denote the set defined in the above lemma. The fact that $E\subseteq \Cmat$ follows simply by setting $U_2=U_2'$ and noticing that (\ref{eq:achievable_start})-(\ref{eq:achievable_end}) imply (\ref{eq:achievable_alter_start}) and (\ref{eq:achievable_alter_end}). To prove that $\Cmat\subset E$, we first note that for a given $p(qu_2x_1x_2)$, $\Cmat$ is a pentagon with two extreme points in the first quadrant given by
\bqa
p_1&=&\left(I(X_1;Y_1|U_2,Q=q), I(U_2;Y_1|Q=q)+I(X_2;Y_2|U_2,Q=q)\right),\\
p_2&=&\left(I(U_2X_1;Y_1|Q=q)-I(U_2;Y_2|Q=q), I(X_2;Y_2|Q=q)\right).
\eqa
It suffices to show that, for any given $p(qu_2x_1x_2)$ in (\ref{eq:input_dist}), the corresponding $p_1$ and $p_2$, belongs to the set $E$, where

That $p_1\in E$ follows from setting $U_2=U_2'$. To show that $p_2\in E$, we use the following inequality
\bqn
I(U_2X_1;Y_1|Q=q)-I(U_2;Y_2|Q=q)&=&I(U_2;Y_1|X_1Q=q)-I(U_2;Y_2|Q=q)+I(X_1;Y_1|Q=q)\\&\leq& I(X_1;Y_1|Q=q)\\&\leq& I(X_1;Y_1|U_2,Q=q).
\eqn
Hence, $\Cmat\subseteq E$.
\end{proof}

\subsection{Capacity Outer Bound for DMZIC with Weak Interference}\label{sec:equivalence}
Costa proved in \cite{Costa:85IT} that a GZIC with weak interference is equivalent in capacity region to a degraded GIC. As such, Sato's outer-bound on degraded GIC \cite{Sato:78IT2} applies to that of the GZIC with weak interference. Sato's outer-bound is in essence the capacity region of a related Gaussian broadcast channel, which is a natural outer-bound to the interference channel due to its implied transmitter cooperation. In this section, we use the same technique to obtain a capacity outer-bound for DMZIC with weak interference, i.e., that satisfies the Markov chain $X_2-(X_1,Y_2)-Y_1$. Specifically, for any such DMZIC with weak interference, one can find an equivalent (in capacity region) DMDIC whose capacity region is bounded by that of an associated degraded broadcast channel.

\begin{theorem}\label{thm:outer-bound}
For a DMZIC that satisfies the Markov chain $X_2-X_1Y_2-Y_1$, the capacity region is outer-bounded by
\small
\bqn
\Rmat_{OB}=\overline{co}\left\{\bigcup_{p(u)p(x_1x_2|u)}(R_1,R_2)\left|\begin{array}{l}R_1\leq I(U;Y_1),\\ R_2\leq I(X_1X_2;Y_2'|U)\end{array}\right.\right\},
\eqn
\normalsize
where $U-X_1X_2-Y_2'-Y_1$ forms a Markov chain and $\|\Umat\|=\min\{\|\Ymat_1\|,\|\Ymat_2'\|,\|\Xmat_1\|\cdot\|\Xmat_2\|\}$, and $\overline{co}\left\{\cdot\right\}$ denotes the closure of the convex hull operation.
\end{theorem}
\begin{proof}
Suppose that the DMZIC with weak interference has inputs $X_1$, $X_2$ and outputs $Y_1$, $Y_2$ respectively. Let us denote by  $X_1'$, $X_2'$ and $Y_1'$, $Y_2'$ the inputs and outputs of another DMIC with $X_1'=X_1$, $X_2'=X_2$, $Y_1'=Y_1$, and $Y_2'$ to be a function of $X_1$ and $Y_2$, denoted as $Y_2'=f(X_1,Y_2)$ such that the Markov chain $(X_1',X_2')-Y_2'-Y_1'$ holds. Thus, the DMIC specified by the input pair $(X_1',X_2')$, and the output pair $(Y_1',Y_2')$ is indeed a DMDIC.
	
The proof that this DMDIC has the same capacity region as the specified DMZIC, and hence is outer-bounded by the associated broadcast channel  follows in exactly the same fashion as Costa's proof for the Gaussian case \cite{Costa:85IT}, hence is omitted here.
\end{proof}

\textit{Remark $3$:} A trivial choice of $Y_2'$ is a bijection of $X_1$ and $Y_2$. It is easy to verify that the Markov chain $(X_1',X_2')-Y_2'-Y_1'$ holds for such $Y_2'$. However, other $Y_2'$ can be constructed, as long as the Markov chain $(X_1',X_2')-Y_2'-Y_1'$ is satisfied. Nevertheless, the associated broadcast channels would have the same the capacity region.

\subsection{Examples}
\begin{example}Consider a DMZIC with input and output alphabets $\Xmat_1 = \Xmat_2 = \Ymat_1 = \Ymat_2 = \{0,1\}$ and is defined by the equations: $y_1 = x_1\cdot x_2$, $y_2 = x_2$. Etkin and Ordentlich in \cite{Etkin&Ordentlich:2011ISIT} established the capacity region for this binary multiplier channel via a new outer-bounds derived in their paper. As this channel satisfies the weak interference condition in this paper, we can immediately get the sum-rate capacity to be $\displaystyle\max_{p(x_1)p(x_2)}I(X_1;Y_1)+I(X_2;Y_2)$.
\end{example}

\begin{example}\label{ex:2}
Let $\Xmat_1=\Xmat_2=\Ymat_1=\Ymat_2=\{0,1\}$ and
\bqn
Y_1&=&X_1\oplus Y_2,\\
Y_2&=&X_2\oplus Z,
\eqn
where $\oplus$ denotes the modulo $2$ sum and $Z\sim \mbox{Bern}(\epsilon)$.

Clearly, the Markov chain $X_2-X_1Y_2-Y_1$ is satisfied. Let $p=\Pr(X_2=1)$.
Then,
\bqn
I(X_2;Y_2)&=&h_2\left(e(1-p)+(1-\epsilon)p\right)-h_2(\epsilon),\\
I(X_1;Y_1)&=&H(Y_1)-h_2\left(\epsilon(1-p)+(1-\epsilon)p\right).
\eqn
The sum-rate capacity is
\bqn
C_{\mbox{sum}}=\max_{p(x_1)p(x_2)} \{I(X_1;Y_1)+I(X_2;Y_2)\} = 1-h_2(\epsilon),
\eqn
which is achieved by any $p(x_1)p(x_2)$ such that $H(Y_1)=1$. Additionally, both points $(0,1-h_2(\epsilon))$ and $(1-h_2(\epsilon),0)$ are trivially achievable. Therefore, the capacity region of this channel is the triangle connecting the two rate pairs $(0,1-h_2(\epsilon)$ and $(1-h_2(\epsilon),0)$.

This channel does not belong to any class of channels that have been studied in the literature. The property of $H(Y_1|X_1)=H(Y_2)$ is similar to the deterministic interference channel definition \cite{ElGamal&Costa:82IT}. However, $Y_2$ is not a deterministic function of $X_2$.

This channel is equivalent, in the capacity region, to the following interference channel:
\bqn
Y_1&=&X_1\oplus X_2 \oplus Z,\\
Y_2&=&X_1\oplus X_2 \oplus Z.
\eqn
This can be proved in a similar way to that used in \cite{Costa:85IT} for proving the equivalence between the Gaussian ZIC and the Gaussian degraded IC.  Notice that the capacity region of the discrete additive degraded IC is solved by Benzel in \cite{Benzel:79IT}, the capacity region of the DMZIC can be obtained through the equivalent discrete additive degraded IC, i.e., the closure of the convex hull of all the nonnegative $(R_1,R_2)$ satisfying the following inequalities:
\bqn
R_1\leq I(X_1;Y_1),\\
R_2\leq I(X_2;Y_2),
\eqn
for all possible product input distribution on $\Xmat_1\times \Xmat_2$.
\end{example}

\begin{example}
Let $\Xmat_1=\Xmat_2=\Ymat_1=\Ymat_2=\{0,1\}$ and
\bqn
Y_1&=&X_1\cdot Y_2,\\
Y_2&=&X_2\oplus Z.
\eqn
This channel is similar to Example \ref{ex:2} except that $Y_1$ is replaced by an erasure channel.

The Markov chain $X_2-X_1Y_2-Y_1$ holds and the capacity region of this channel can be obtained in a manner similar to that of \cite{Etkin&Ordentlich:2011ISIT}.
We first upper-bound the two individual rates $R_1$ and $R_2$. From the proof of Theorem \ref{thm:DMZIC_CR}, it is straightforward to obtain
\bqn
R_1-\epsilon_1\leq I(UX_1;Y_1|Q)
\eqn
where $U$ is an auxiliary random variable satisfying $p(ux_1x_2) = p(x_1)p(ux_2)$. For $R_2$,
\bqn
n(R_2-\epsilon_2)&\leq& I(X_2^n;Y_2^n)\\
&\leq&\sum_{i=1}^n\left(H(Y_{2i}|Y_2^{i-1})-H(Y_{2i}|X_2^nY_2^{i-1})\right)\\
&\leq&\sum_{i=1}^n\left(H(Y_{2i})-H(Y_{2i}|X_{2i})\right)\\
&=&\sum_{i=1}^nI(X_{2i};Y_{2i})\\
&=&nI(X_{2};Y_{2}|Q).
\eqn
Let $p_{1,q}=\Pr(x_1=1|Q=q)$, $p_{2,q}=\Pr(x_2=1|Q=q)$, $p_{2,q}^y=\Pr(y_2=1|Q=q)$, $r_q=H(Y_2|U,q)$, note that
\bqn
p_{2,q}^y=p_{2,q}(1-\epsilon)+(1-p_{2,q})\epsilon,
\eqn
and
\bqn
r_q\leq h_2(p_{2,q}),
\eqn
for each $q$. Then,
\bqn
R_1-\epsilon_1&\leq& I(UX_1;Y_1|Q)\\
&=&\sum_{q=1}^{\|\Qmat\|}[H(Y_1|q)-\sum_{x_1=0}^1p(x_1|q)H(Y_1|x_1,U,q)]\\
&=&\sum_{q=1}^{\|\Qmat\|}[H(Y_1|q)-p(x_1=1|q)H(Y_2|U,q)]\\
&=&\sum_{q=1}^{\|\Qmat\|}[h_2(p_{1,q}p_{2,q}^y|q)-p(x_1=1|q)r_q]
\eqn
and
\bqn
R_2-\epsilon_2&\leq& I(X_2;Y_2|Q)\\
&=&H(Y_2|Q)-H(Y_2|X_2Q)\\
&=&h_2(p_{2,q}^y)-h_2(\epsilon).
\eqn
Compared with the expressions in \cite[Eqs.~(15) and (16)]{Etkin&Ordentlich:2011ISIT}, the only difference is the constant $h_2(\epsilon)$, which does not affect the optimization. Therefore, the optimization process there can be directly applied here. It follows that the capacity region of this channel is the convex hull of $\Rmat^\prime$, where
\bqn
\Rmat^\prime = \bigcup_{0\leq p_1,p_2\leq 1}\left\{(R_1,R_2):R_1\leq I(X_1;Y_1)=h_2(p_1p_{y_2})-p_1h_2(p_{y_2}),R_2\leq I(X_2;Y_2)=h_2(p_{y_2})-h_2(\epsilon)\right\},
\eqn
where $p_{y_2}=\epsilon(1-p_2)+(1-\epsilon)p_2$. Clearly, the sum-rate capacity is $\displaystyle\max_{p_1p_2}\left\{(p_1p_{y_2})+(1-p_1)h_2(p_{y_2})-h_2(\epsilon)\right\}$.
\end{example}
\begin{example}
$\|\Xmat_1\|=\|\Xmat_2\|=\|\Ymat_2\|=2$, $\|\Ymat_1\|=3$.
\bqn
Y_1&=&\left\{\begin{array}{cl}X_1\oplus Y_2,& \textrm{with probability } 1-\delta\\ e, & \textrm{with probability } \delta\end{array},\right.\\
Y_2&=&X_2\oplus V_1,
\eqn
where $V_1\sim \mbox{Bern}(\epsilon)$. Clearly, $Y_1$ is the output of a erasure channel with input $X_1\oplus Y_2$ and erasure proability $\delta$.
Define $Y_2'=X_1\oplus Y_2$. Thus, the DMIC with inputs $X_1, X_2$ and outputs $Y_1,Y_2'$ is a degraded DMIC. The capacity region of this degraded DMIC has been solved by Liu and Ulukus \cite{Liu&Ulukus:08IT}, and can be expressed as
\bqn
\Rmat_I = \overline{co}\left\{\bigcup_{p(x_1)p(x_2)}\left((R_1,R_2):R_1\leq I(X_1;Y_1),R_2\leq I(X_2;Y_2'|X_1)\right)\right\}.
\eqn
The corresponding capacity region for the DMZIC is 
\bqn
\Rmat_Z=\overline{co}\left\{\bigcup_{p(x_1)p(x_2)}\left((R_1,R_2):R_1\leq I(X_1;Y_1),R_2\leq I(X_2;Y_2)\right)\right\}.
\eqn
\end{example}
That $\Rmat_Z$ being the capacity region comes from the fact that $I(X_2;Y_2'|X_1)=I(X_2;Y_2)$ while $\Rmat_I$ is naturally an outer-bound. 
\begin{example}\label{ex:5}
Let $\|\Xmat_1\|=\|\Xmat_2\|=\|\Ymat_1\|=\|\Ymat_2\|=2$ and the channel transition probability be given by
\bqn
p(y_1y_2|x_1x_2)=p(y_2|x_2)p(y_1|x_1y_2),
\eqn
where $p(y_2|x_2)$ and $p(y_1|x_1y_2)$ are specified in Table \ref{table:ex1}.
\begin{table}[htp]
\centering
\caption{Channel Transition Probabilities}\label{table:ex1}
		\begin{tabular}{|c|c|c||c|c|c|}
  \hline
   $p(y_2|x_2)$ & $y_2=0$ & $y_2=1$ & $p(y_1|x_1y_2)$ & $y_1=0$ & $y_1=1$\\\hline
	 $x_2=0$ & $.1$ & $.9$ & $x_1y_2=00$ or $11$ & $.75$ & $.25$\\\hline
	 $x_2=1$ & $.9$ & $.1$ & $x_1y_2=01$ or $10$ & $0$ & $1$\\\hline
		\end{tabular}
\end{table}

By Theorem \ref{thm:DMZIC_CR}, the sum-rate capacity is
\bqn
\Cmat_{sum} = \max_{p(x_1)p(x_2)} I(X_1;Y_1)+I(X_2;Y_2)\approx.531.
\eqn
In addition, a simple outer bound can be constructed as follows
\bqa
R_1&\leq&I(X_1;Y_1|X_2),\label{eq:simple_ob_start}\\
R_2&\leq&I(X_2;Y_2),\\
R_1+R_2&\leq&I(X_1;Y_1)+I(X_2;Y_2).\label{eq:simple_ob_end}
\eqa

We now use Theorem \ref{thm:outer-bound} to obtain a new outer bound. Construct $Y_2'$ as follows
\bqn
Y_2'=\left\{\begin{array}{l}0,\textrm{     if }x_1y_2 = 00\textrm{ or }11,\\1, \textrm{   otherwise}.\end{array}\right.
\eqn
Then $p(y_2'|x_1x_2)$ is given in Table~\ref{table:y_2'}.
\begin{table}[htp]
\centering
\caption{$P(Y_2'|X_1X_2)$}\label{table:y_2'}
\begin{tabular}{|c|c|c|}
\hline
$p(y_2'|x_1x_2)$ & $y_2'=0$ & $y_2'=1$\\\hline
$x_1x_2=00$ & $.1$ & $.9$\\\hline
$x_1x_2=01$ & $.9$ & $.1$\\\hline
$x_1x_2=10$ & $.9$ & $.1$\\\hline
$x_1x_2=11$ & $.1$ & $.9$\\\hline
\end{tabular}
\end{table}

Using Theorem \ref{thm:outer-bound}, the capacity region of the DMZIC is outer-bounded by  that of the associated discrete memoryless degraded broadcast channel:
\bqn
\Rmat_{OB}=\overline{co}\left\{\bigcup_{p(u)p(x_1x_2|u)}(R_1,R_2)\left|\begin{array}{l}R_1\leq I(U;Y_1),\\ R_2\leq I(X_1X_2;Y_2'|U)\end{array}\right.\right\},
\eqn
Let $R_2$ to be fixed at $x$, then
\bqn
\max_{R_2=x}R_1 &=& \max_{H(Y_2'|U)=x+h_2(.1)} H(Y_1)-H(Y_1|U)\\
&\leq& \log(|\Ymat_1|)-f_T(x+h_2(.1)),
\eqn
where $f_T(\cdot)$ is a function defined by Witsenhausen and Wyner \cite{Witsenhausen&Wyner:75IT}.
Fig.~\ref{fig:outer-bound} depicts the new outer-bound specified by
\bqa
\Rmat_{OB}'=\left\{(R_1,R_2)|R_1\leq \log|\Ymat_1|-f_T(x+h_2(.1)), R_2\leq x\right\}.
\eqa
This new outer-bound significantly improves upon the simple outer-bound (\ref{eq:simple_ob_start})-(\ref{eq:simple_ob_end}).

\begin{figure}[htp]
\centerline{\begin{psfrags} \psfrag{R1}[c]{$R_1$}
\psfrag{R2}[c]{$R_2$} \scalefig{.75}\epsfbox{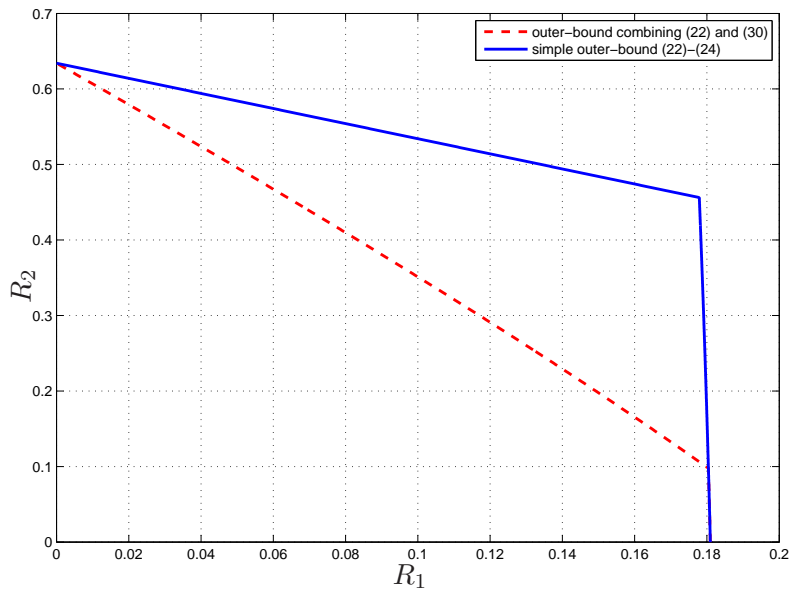}
\end{psfrags}}
\caption{\label{fig:outer-bound} Comparison of the outer-bounds.}
\end{figure}
\end{example}

\section{The DMIC with Mixed Interference}\label{sec:mixed}
For the GIC with mixed interference ($a\leq 1$ and $b\geq 1$ in (\ref{eq:GIC_y1}) and (\ref{eq:GIC_y2})), one can construct an equivalent GIC with degradedness defined by the Markov chain $X_2-(X_1,Y_2)-Y_1$:
\bqn
Y_1'&=&(1-ab)X_1+aY_2+Z_1',\\
Y_2&=&bX_1+X_2+Z_2,
\eqn
where $Z_1'\sim \Nmat(0,1-a^2)$.
This motivates us to define DMIC with mixed interference in an analogous fashion, which leads directly to its sum-rate capacity described in Theorem \ref{thm:mixed_cr}.
\begin{definition}
A DMIC is said to have \textit{mixed interference} if it satisfies the Markov chain
\small
\bqa\label{eq:mixed_cond1}
X_2-(X_1,Y_2)-Y_1
\eqa 
\normalsize and
\small
\bqa\label{eq:mixed_cond2}
I(X_1;Y_1|X_2)\leq I(X_1;Y_2|X_2)
\eqa
\normalsize
for all possible product distributions on $\Xmat_1\times\Xmat_2$.
\end{definition}

\begin{theorem}\label{thm:mixed_cr}
The sum-rate capacity of a DMIC with mixed interference, i.e., one that satisfies the two conditions (\ref{eq:mixed_cond1}) and (\ref{eq:mixed_cond2}), is
\small
\bqa\label{eq:dmic_mixed_sc}
C_{\mbox{sum}}=\max_{p(x_1)p(x_2)}\left\{I(X_2;Y_2|X_1)+\min\{I(X_1;Y_1),I(X_1;Y_2)\}\right\}.
\eqa
\end{theorem}
\normalsize
\begin{proof}
In order to achieve this sum rate, user $1$ transmits its message at a rate such that both receivers can decode it by treating the signal from user $2$ as noise; user $2$ transmits at the interference-free rate since receiver $2$ is able to subtract the interference from user $X_1$.

For the converse, we prove the following two sum-rate bounds separately:
\small
\bqa
n(R_1+R_2)&\leq& \sum_{i=1}^nI(X_{1i}X_{2i};Y_{2i}),\label{eq:mixed_sc1}\\
n(R_1+R_2)&\leq& \sum_{i=1}^nI(X_{1i};Y_{1i})+I(X_{2i};Y_{2i}|X_{1i}).\label{eq:mixed_sc2}
\eqa
\normalsize
For (\ref{eq:mixed_sc1}), the derivation follows the same steps as Costa and El Gamal's result\cite{Costa&ElGamal:87IT}.
For (\ref{eq:mixed_sc2}), we apply similar techniques used in the proof of Theorem \ref{thm:DMZIC_CR}. First, notice that (\ref{eq:mixed_cond1}) implies
\small
\bqa
I(U;Y_1|X_1)\leq I(U;Y_2|X_1)\label{eq:mixed_cond3}
\eqa
\normalsize
for any $U$ whose joint distribution with $X_1,X_2,Y_1,Y_2$ is
\small
\bqa
p(u,x_1,x_2,y_1,y_2)=p(u)p(x_1x_2|u)p(y_1y_2|x_1x_2). \label{eq:u}
\eqa 
\normalsize Therefore,
\small
\bqn
n(R_1+R_2)-n \epsilon &\stackrel{(a)}{\leq}&\!\!\!\!\!\!I(X_1^n;Y_1^n)+I(X_2^n;Y_2^n|X_1^n)\\
&=&\!\!\!\!\!\!\sum_{i=1}^n \left(H(Y_{1i}|Y_1^{i-1})-H(Y_{1i}|Y_1^{i-1}X_1^n)+H(Y_{2i}|Y_2^{i-1}X_1^n)-H(Y_{2i}|Y_2^{i-1}X_2^nX_1^n)\right)\\
&\stackrel{(b)}{\leq}&\!\!\!\!\!\!\sum_{i=1}^n \left(H(Y_{1i})-H(Y_{1i}|Y_1^{i-1}X_1^nY_2^{i-1})+H(Y_{2i}|U_iX_{1i})-H(Y_{2i}|X_{2i}X_{1i}U_i)\right)\\
&=&\!\!\!\!\!\!\sum_{i=1}^n (I(U_iX_{1i};Y_{1i})+I(X_{2i};Y_{2i}|U_iX_{1i}))\\
&=&\!\!\!\!\!\!\sum_{i=1}^n (I(X_{1i};Y_{1i})+I(U_i;Y_{1i}|X_{1i})+I(X_{2i};Y_{2i}|U_iX_{1i}))\\
&\stackrel{(c)}{\leq}&\!\!\!\!\!\!\sum_{i=1}^n (I(X_{1i};Y_{1i})+I(U_i;Y_{2i}|X_{1i})+I(X_{2i};Y_{2i}|U_iX_{1i}))\\
&\stackrel{(d)}{=}&\!\!\!\!\!\!\sum_{i=1}^n (I(X_{1i};Y_{1i})+I(X_{2i};Y_{2i}|X_{1i})),
\eqn
\normalsize
where $(a)$ is because of the independence between $X_1^n$ and $X_2^n$; $(b)$ is from the fact that conditioning reduces entropy and by defining $U_i\triangleq (X_1^{i-1}X_{1,i+1}^n, Y_2^{i-1})$; $(c)$ is from (\ref{eq:mixed_cond3}); and $(d)$ is because of the memoryless property of the channel and (\ref{eq:u}).
From (\ref{eq:mixed_sc1}) and (\ref{eq:mixed_sc2}), we have
\small
\begin{eqnarray}
R_1+R_2&\leq& \sum_{i=1}^n{\min\{I(X_{1i}X_{2i};Y_{2i}), I(X_{1i};Y_{1i})+I(X_{2i};Y_{2i}|X_{1i})\}}.
\end{eqnarray}
\normalsize
Finnally, by introducing the time-sharing  random variable $Q$, one obtains (\ref{eq:dmic_mixed_sc}) as desired.
\end{proof}

We give the following example where the obtained sum-rate capacity helps determine the capacity region of a DMIC.

\begin{example}
	Consider the following deterministic channel:
	\small
	\bqn
	Y_1 &=& X_1\cdot X_2,\\
	Y_2 &=& X_1 \oplus X_2,
	\eqn
	\normalsize
	where the input and output alphabets $\Xmat_1=\Xmat_2=\Ymat_1=\Ymat_2 = \{0,1\}$. Notice that this channel does not satisfy the condition of the deterministic interference channel in \cite{ElGamal&Costa:82IT}.
	Obviously, the Markov chain (\ref{eq:mixed_cond1}) holds. Moreover,
	\small
	\bqn
	I(X_1;Y_1|X_2)\!\!\!\!\!\!&=&\!\!\!\!\!\!H(Y_1|X_2)= p(x_2=1)H(X_1),\\
	I(X_1;Y_2|X_2)\!\!\!\!\!\!&=&\!\!\!\!\!\!H(Y_2|X_2) = H(X_1).
	\eqn
	\normalsize
	Therefore,
	\small
	\bqn
	I(X_1;Y_1|X_2)\leq I(X_1;Y_2|X_2),
	\eqn
	\normalsize
	for all possible input product distributions on $\Xmat_1\times\Xmat_2$.
	Thus, this is a DMIC with mixed interference. On applying Theorem \ref{thm:mixed_cr}, we compute the sum-rate capacity to be
	\small
	\bqn
	\Cmat_{\mbox{sum}}\!\!&=&\!\!\max_{p(x_1)p(x_2)}\left[\min(I(X_1X_2;Y_2),I(X_1;Y_1)+I(X_2;Y_2|X_1))\right]\\
	&=&1.
	\eqn
\normalsize
Given that $(1,0)$ and $(0,1)$ are both trivially achievable, the above sum-rate capacity leads to the capacity region for this DMIC to be $\{(R_1,R_2):R_1+R_2\leq 1\}$.
\end{example}
\section{Conclusion}\label{sec:conclusion}
In this paper, we have derived the sum-rate capacity for a class of discrete memoryless interference channels whose channel property resembles that of the Gaussian interference channel with one-sided and weak interference. Capacity outer bounds are also derived for this class of channels. The same technique is then applied to obtain the sum-rate capacity of discrete memoryless interference channels with mixed interference. For both cases, the capacity expressions as well as the encoding schemes that achieve the sum-rate capacity are analogous to the Gaussian interference channel counterpart. These results allow us to obtain capacity results for several new discrete memoryless interference channels.

\begin{appendix}
\subsection{Counter Example for the Equivalence between the Two Different Conditions}\label{apx:example}
This example explains that a DMZIC that satisfies the mutual information condition (\ref{eq:DMZIC_weakextention}) does not necessarily imply the Markov chain relationship (\ref{eq:weakcond2}).

Let $f_{ij}$ represent $p(y_1=1|x_1=i, x_2=j)$, $g_{j}$ represent $p(y_2=1|x_2=j)$, $p_i = \Pr\{X_i=1\}$, and $\bar{p}_i=1-p_i$ $(i,j\in \{0,1\})$.
From the mutual information condition (\ref{eq:DMZIC_weakextention})
\bqn
I(X_2;Y_2)&\geq& I(X_2;Y_1|X_1),
\eqn
we have
\bqn
H(Y_2)-H(Y_2|X_2) &\geq& H(Y_1|X_1)-H(Y_1|X_1,X_2)\\
h_2(\bar{p}_2g_0+p_2g_1)-p_2h_2(g_1)-\bar{p}_2h_2(g_0) &\geq& \bar{p}_1h_2(\bar{p}_2f_{00}+p_2f_{01})+p_1h_2(\bar{p}_2f_{10}+p_2f_{11})\\&&-\bar{p}_1\bar{p}_2h_2(f_{00})-\bar{p}_1p_2h_2(f_{01})-p_1\bar{p}_2h_2(f_{10})-p_1p_2h_2(f_{11})
\eqn
Upon obtaining the above inequality, one can make specific choices of $\{f_{ij}\}$ and $\{g_{j}\}$ to make the above inequality hold for all possible $p_1$ and $p_2$ range from $0$ to $1$. For example, it is easy to verify that a valid choice is
\bqn
f_{00}&=&.1 , f_{01}=.3, f_{10} = .5, f_{11} = .25,\\
g_0 &=& .1, g_1 = .5.
\eqn

In the following, we prove by contradiction that this channel does not satisfy the markov chain condition (\ref{eq:weakcond2}).

Suppose that the markov chain (\ref{eq:weakcond2}) is satisfied,
\bqn
p(y_1|x_1x_2y_2) = p(y_1|x_1y_2).
\eqn
Then we would have,
\bqn
p(y_1|x_1x_2)=\sum_{y_2}p(y_1y_2|x_1x_2)=\sum_{y_2}p(y_2|x_2)p(y_1|x_1y_2).
\eqn
Solving this equation, we get
\bqn
p(y_1=1|x_1=1,y_2=1)=-\frac{1}{16},
\eqn
which contradicts the fact that channel transit probability can never be negative.
\end{appendix}

\end{document}